\documentclass[11pt,a4paper]{article} 
\usepackage{epstopdf,mathrsfs,tabularx}	
\usepackage{amsmath, amssymb, graphics, setspace}
\usepackage{epstopdf}
\usepackage[utf8]{inputenc} 
\usepackage{amsmath}
\usepackage{amsfonts}
\usepackage{multirow}
\usepackage{amssymb}
\usepackage{amsthm}

\usepackage{geometry} 
\usepackage{graphicx} 

\usepackage{booktabs} 
\usepackage{array} 
\usepackage{paralist} 
\usepackage{verbatim} 
\usepackage{subfig} 
\usepackage{amsmath,amssymb,amsthm}
\usepackage{fancyhdr} 
\usepackage{lipsum}
\usepackage[titles,subfigure]{tocloft} 

\newtheorem{theorem}{Theorem}

\newtheorem{definition}{Definition}
\newtheorem{remark}{Remark}

\title{On Quantile curves based bivariate reliability concepts}
\author{{ Sreelakshmi N }\\
\small{\it Indian Statistical Institute, Chennai, India }}
\date{}
\begin{document}
\doublespace
\maketitle
\begin{abstract}
We extend the univariate quantile based reliability concepts to the bivariate case using quantile curves. We propose quantile curves based bivariate hazard rate and bivariate mean residual life function and  establish a relationship between them. We study the uniqueness properties of these concepts to determine the underlying quantile curve. We also study the quantile curves based reliability concepts in reverse time. \\

\end{abstract}

\section{Introduction}
Several generalizations of univariate reliability concepts to bivariate as well as multivariate setup is available in the literature. We use quantile curves to define bivariate hazard rate and bivariate mean residual life function.   Let  $X$ be random variable with  distribution function $F(.)$. The quantile function is defined as
\begin{equation*}
  Q_{X}(u)= \inf \{x:F(x) \ge u\};\,\,0 \le u \le 1.
\end{equation*}
When $F$ is absolutely continuous $Q_{X}(u)=F^{-1}(u)$.
The properties of quantile functions are well explained in Gilchrist (2000). It is not possible to extend the notion of univariate quantile functions to multivariate set up in a unique manner. Such extensions can result some issues like non existence of natural ordering in $n$ dimensions and non-parametric estimation of multivariate quantiles.
\par Several definitions are avaialble for multivaraite quantiles such as median balls (Av{\'e}rous and Meste, 1997), zonoid quantile (Koshevoy and Mosler, 1997), the concept of half planes (Nola  (1992) and Massé and Theodorescu (1994)) and the notion of depth function (Tukey, 1977).

\par   Fernández-Ponce and Suarez-Llorens (2003)   fixed the problem related to the ordering in $n$ dimensions as well as the choice of shape of central region in the case of no symmetric distributions. They defined the multivariate quantile as a set of points known as quantile curves. Quantile curves accumalate the same probability for a fixed orthant.

\par Let $\textbf{X}=\left( {X ,Y } \right)$ be absolutely continuous  bivariate random vector and $\textbf{x}=\left( {x ,y } \right)$ be a point in $R^2$. Let $F(.)$ and $G(.)$ be the marginal distribution functions of $X$ and $Y$, respectively. Also let $F_{\varepsilon}(\textbf{x})$ be the joint distribution function of $\textbf{X}$.  Denote the four directions in the two dimensional plane as $\varepsilon  = \left( {\varepsilon _1 ,\varepsilon _2 } \right)$ with $\varepsilon _i  \in \left\{ { - 1,1} \right\};i = 1,2.$  We use the same notation  $-$ and $+ $ given in Belzunce et al. (2007) to represent  $-1$ and $+1$, respectively. We have,
$$F_{\varepsilon _{ -  - } } \left( \textbf{x} \right) = P\left[ {X  \le x ,Y  \le y } \right].$$
 The $p-{th}$ bivariate quantile curve for the direction $\varepsilon$ denoted by $Q_\textbf{X} (p,\varepsilon )$ defined as
$$Q_\textbf{X} (p,\varepsilon ) = \left\{ {\left( \textbf{x} \right) \in R^2 :F_\varepsilon  \left( \textbf{x} \right) = p} \right\},$$
where $0 \le p \le 1$.
\par As an example, for the random vector $\textbf{X}$ having  bivariate Pareto distribution with independent components following  univariate Pareto distribution with scale parameter $k$ and shape parameter $\alpha$, it is possible to plot four quantile curves in four directions for each $p \in [0,1]$.
For the direction $\varepsilon _{ +  + } $ we obatin the quantile curve by plotting $xy=k^{2}p^{-1/\alpha}$ for $0 \le p \le 1.$ {\bf please plot the figure}
\par Suppose $F(x)=u$ and $G(y)=v$, then  $x  = F^{ - 1} (u) = Q_{X}(u)$ and $y  = G^{ - 1} (v) = Q_{Y}(v)$.
For the bivariate random vector $\textbf{X}$, Belzunce et al. (2007) showed that the quantile curves can be expressed by means of quantiles of the conditional distributions of  $Y|X \le x$ and $Y|X \ge x$ given as
\begin{eqnarray}\label{qcr}
 Q_{\textbf{X}} (p,\varepsilon _{ - \, - } ) \to \left\{ {\left( {Q_{X} (u),Q_{Y|X \le Q_{X} (u)} (\frac{p}{u})} \right):u > p} \right\} \nonumber\\
 Q_{\textbf{X}} (p,\varepsilon _{ + \, - } ) \to \left\{ {\left( {Q_{X} (u),Q_{Y|X \ge Q_{X} (u)} (\frac{p}{{1 - u}})} \right):u < 1 - p} \right\} \\
 Q_{\textbf{X}} (p,\varepsilon _{ - \, + } ) \to \left\{ {\left( {Q_{X} (u),Q_{Y|X \le Q_{X} (u)} (1 - \frac{p}{u})} \right):u > p} \right\} \nonumber\\
 Q_{\textbf{X}} (p,\varepsilon _{ + \, + } ) \to \left\{ {\left( {Q_{X} (u),Q_{Y|X \ge Q_{X} (u)} (1 - \frac{p}{{1 - u}})} \right):u < 1 - p} \right\}. \nonumber
 \end{eqnarray}
where $0 \le u,p \le 1$.
\par  Nair et al. (2013) studied the reliability concepts in a univariate quantile frame work. Sreelakshmi (2017) introduced the bivariate reliability concepts using  dependence structure and studied properties as well as characterizations  based on the relationship between copula based bivariate hazard rate and bivariate mean residual life. For each direction, the probability $F_\varepsilon  \left( . \right)$ is written in terms of copula and thus lead to the construction of bivariate reliability concepts. This approach is useful to study the reliability properties of copula based models. In this paper, we use the idea developed by Belzunce et al. (2007) given in (\ref{qcr}) to represent the bivariate reliability concepts using univariate quantiles as well as quantiles of conditional distribution. It is interesting to note that this way of expressing the level curve does   not require the concept of copula and it requires only quantiles of conditional distributions.
\par  The article is organised as follows. Section 2 gives the definitions on quantile curves based bivariate hazard rate and bivariate mean residual life function. The uniqueness properties and the relationships between these quantile curves based reliability concepts  are derived in Section 3. In Section 4,  we introduce  the corresponding reliability notions in reversed time setup. 
In Section 5 we give the conclusion of our study.

\section{Quantile curves based bivariate reliability concepts}
\par To introduce the bivariate reliability concepts based on quantile curves, we need some notations which are explained below.
Let $\textbf{X}=(X, Y)$ be the bivariate random vector with univariate marginal distribution function $F(.)$ and $G(.)$. Suppose that $F(x)=u$ so that $x=F^{-1}(u)=Q_{X}(u)$ and  $G(y)=v$ so that $y=G^{-1}(v)=Q_{Y}(v)$. Let $Y_{X}=(Y|X \le x)$ and $G_{1}(Y_X)=P(Y \le y|X \le x)=P_{X}$. Therefore, $Y_{X}=G_{1}^{-1}(P_{X})=\phi(P_{X})$. Also $X_{Y}=X|Y \le y$ and $G_{2}(X_Y)=P(X \le x|Y \le y)=P_{Y}$ which yields $X_{Y}=G_{2}^{-1}(P_{Y})=\psi(P_{Y})$.
\par Next, we propose the definitions for bivariate hazard rate as well as bivariate mean residual life based on quantile curves for the direction $\varepsilon _{ -  - } $. From (1), the quantile curve in the direction $\varepsilon _{ -  - } $, $Q_{\textbf{X}}(p,\varepsilon _{ -  - })$ is a vector containing $u-{th}$ quantile corresponding to the random variable $X$, $Q_{X}(u)$ and $Q_{Y|X \le Q_{X}(u)}(p/u)$ which is the quantile function corresponds to the random variable $Y_{X}$ and  a point that accumulates probability $\frac{p}{u}$ to the left tail and $1-\frac{p}{u}$ to the right tail provided $u>p$. Based on the above vector of two quantile functions, we look into the definitions of quantile curves based bivariate reliability concepts.

\begin{definition}
The bivariate hazard rate of $\textbf{X}$ in the direction $\varepsilon _{ -  - } $ in terms of quantile curves can be defined as the vector
\begin{equation}
\underset{\raise0.3em\hbox{$\smash{\scriptscriptstyle-}$}}{h} _{\varepsilon _{ -  - } } (u,P_{X}) = \left( {h_1 (u),h_2 (P_{X})} \right)
\end{equation}
where
\begin{equation}
h_1 (u)=\frac{1}{(1-u)Q'_{X}(u)}
\end{equation}
and
\begin{equation}
h_2 (P_{X})=\frac{1}{(1-P_{X})\phi'(P_{X})},
\end{equation}
where prime denotes the differentiation with respect to $u$.
  \end{definition}The quantity $h_2 (P_{X})$ can be interpreted as the condional probbaility of the failure of second unit in the next small interval of time given the survival of it at $100(1-P_{X}) \%$ point of distribution and that of first  unit at $100(1-u) \%$ point of distribution.
\par For the last one decade, researchers have shown much interest in studying  the remaining lifetime of a unit given it has survived a particular point of time $t$. For more on bivariate mean residual life based on distribution functions, one can refer to  Nair and Nair (1989) and Kulkarni and Rattihalli (1996).
\begin{definition}The quantile curves based bivariate mean residual life function  of $\textbf{X}$ in  the direction $\varepsilon _{ -  - } $ is given by
\begin{equation}
\underset{\raise0.3em\hbox{$\smash{\scriptscriptstyle-}$}}{m} _{\varepsilon _{ -  - } } (u,P_{X}) = \left( {m_1 (u),m_2 (P_{X})} \right),
\end{equation}
where
\begin{equation}
m_1 (u)= \frac{1}{{1 - u}}\int\limits_u^1 {Q_{X} (z)dz - Q_{X} (u)}
\end{equation}
and
\begin{equation}
m_2 (P_{X})= \frac{1}{{1 - P_{X}}}\int\limits_{P_{X}}^{1} {\phi (z)dz - \phi (P_{X})}.
\end{equation}
\end{definition}
\begin{remark}
If we interchange $X$ and $Y$, we obtain another vector for  the bivariate hazard rate in the direction $\varepsilon _{ -  - } $ and is given by
\begin{equation}
\underset{\raise0.3em\hbox{$\smash{\scriptscriptstyle-}$}}{H} _{\varepsilon _{ -  - } } (v,P_{Y}) = \left( {H_1 (v),H_2 (P_{Y})} \right)
\end{equation}
where
$$H_1 (v)=\frac{1}{(1-v)Q'_{Y}(v)}$$
and
$$H_2 (P_{Y})=\frac{1}{(1-P_{Y})\psi'(P_{Y})}.$$
And the bivariate mean residual life is given by
\begin{equation}
\underset{\raise0.3em\hbox{$\smash{\scriptscriptstyle-}$}}{M} _{\varepsilon _{ -  - } } (v,P_{Y}) = \left( {M_1 (v),M_2 (P_{Y})} \right),
\end{equation}
where
$$M_1 (v)= \frac{1}{{1 - v}}\int\limits_v^1 {Q_{Y} (z)dz - Q_{Y} (v)}$$
and
$$M_2 (P_{Y})= \frac{1}{{1 - P_{Y}}}\int\limits_{P_{Y}}^{1} {\psi (z)dz - \psi (P_{Y})}.$$
\end{remark}
\begin{remark}
If $X$ and $Y$ are independent, $\underset{\raise0.3em\hbox{$\smash{\scriptscriptstyle-}$}}{h} _{\varepsilon _{ -  - } } (u,P_{X})$ reduces to vectors of univariate hazard quantile functions of the random variables $X$ and $Y$ respectively.  That is, when $X$ and $Y$ are independent,
$$\underset{\raise0.3em\hbox{$\smash{\scriptscriptstyle-}$}}{h} _{\varepsilon _{ -  - } } (u,P_{X})=\left(h_{1}(u),H_1({v})\right).$$
Also, when $X$ and $Y$ are independent,
$$\underset{\raise0.3em\hbox{$\smash{\scriptscriptstyle-}$}}{m} _{\varepsilon _{ -  - } } (u,P_{X}) = \left( {m_1 (u),M_1 (v)} \right).$$
\end{remark}
\section{Properties and characterizations}
In this section we study the uniqueness property of the reliability concepts derived in Section 2. We establish a  relationship between the bivariate hazard rate and bivariate mean residual life in the quantile curves based approach.
\begin{theorem}
For the bivariate random vector $\textbf{X}$ , quantile curve based bivariate hazard rate defined by (2) determines the underlying quantile curve $Q_{\textbf{X}} (p,\varepsilon _{ - \, - } )$  uniquely. Here $Q_{\textbf{X}} (p,\varepsilon _{ - \, - } )$  is obtained  as
$$ Q_{\textbf{X}} (p,\varepsilon _{ - \, - } ) = \left( Q_{X} (u),\phi(P_{X}) \right),$$
where
$$Q_{X}(u)=\int_{0}^{u}{\frac{dz}{(1-z)h_{1}(z)}}$$
and
$$\phi(P_{X})=\int_{0}^{P_{X}}{\frac{dz}{(1-z)h_{2}(z)}}.$$
\end{theorem}
\begin{proof}
 Using the definition of  hazard quantile function given in (4), we obtain
$$\phi'(P_{X}) =\frac{1}{(1-P_{X}) h_2 (P_{X})}.$$
Integrating the above equation from $0$  to $ P_{X}$ yields,
 $$\phi(P_{X})=\int_{0}^{P_{X}}{\frac{dz}{(1-z)h_{2}(z)}}.$$
 The expression for $Q_{X}(u)$ can be obtained from $h_{1}(u)$ on similar lines.
\end{proof}
\begin{theorem}
 For the bivariate random vector $\textbf{X}$, the quantile curve  can be expressed uniquely in terms of $\underset{\raise0.3em\hbox{$\smash{\scriptscriptstyle-}$}}{m} _{\varepsilon _{ -  - } } (u,P_{X})$ through
$$ Q_{\textbf{X}} (p,\varepsilon _{ - \, - } ) = \left( Q_{X} (u),\phi(P_{X}) \right),$$
where
\begin{equation}\label{eq10}
  Q_{X}(u)=\mu_{X}-m_{1}(u)+\int_{0}^{u}\frac{m_1(z)}{1-z}dz
\end{equation}
 and
\begin{equation}\label{eq11}
\phi(P_{X})=\mu_{Y_{X}}-m_{2}(P_{X})+\int_{0}^{P_{X}}\frac{m_2(z)}{1-z}dz
\end{equation}
and $\mu_{X}$ and $\mu_{Y_{X}}$ are the means corresponding to the random variables $X$ and $Y_{X}$, respectively.
\end{theorem}
\begin{proof}
We rewrite the definition of $ m_2 (P_{X})$ as
\begin{equation}
m_2 (P_{X})= \frac{1}{{1 - P_{X}}}\int\limits_{P_{X}}^{1} {(1 - z)\phi'(z)dz}.
\end{equation}
Differentiating the above equation with respect to $P_{X}$, we have
$$\phi'(P_{X})=\frac{m_2 (P_{X})}{1-P_{X}}-m_{2}' (P_{X}).$$
On integrating the above equation  from $0$ to $P_X$, we get (\ref{eq11}). Note that $m_{2}(0)=\mu_{Y_{X}}$.  On similar lines, from (6), we can arrive at the form of $Q_{X}(u)$ given in equation (\ref{eq10}).
\end{proof}

\begin{remark}[]
The quantile curve based bivariate hazard rate obtained after interchanging $X$ and $Y$ given in (8) uniquely determines the underlying quantile curve through
$$Q_{Y}(v)=\int_{0}^{v}{\frac{dz}{(1-z)H_{1}(z)}}$$
and
$$\psi(P_{Y})=\int_{0}^{P_{Y}}{\frac{dz}{(1-z)H_{2}(z)}}.$$
Similarly, the bivariate mean residual life given in (9) determines the quantile curve through
$$Q_{Y}(v)=\mu_{Y}-M_{1}(v)+\int_{0}^{v}\frac{M_1(z)}{1-z}dz$$
 and
 $$\psi(P_{Y})=\mu_{X_{Y}}-M_{2}(P_{Y})+\int_{0}^{P_{Y}}\frac{M_2(z)}{1-z}dz,$$
 where $\mu_{Y}$ and $\mu_{X_{Y}}$ are the means of the random variables $Y$ and $X_Y$, respectively. \end{remark}
 \begin{theorem}
   For the bivariate random vector \textbf{X}, the quantile curves based bivariate hazard rate and bivariate  mean residual life function are related through the relationship
  \begin{equation}\label{eq13}
    (1-u)m_1(u)=\int_{u}^{1}\frac{dz}{h_1(z)}
  \end{equation}
   and
  \begin{equation}\label{eq14}
   (1-P_X)m_2(P_X)=\int_{P_X}^{1}\frac{dz}{h_2(z)}.
  \end{equation}
 \end{theorem}
 \begin{proof}
   Substituting (4) in (12), we have the relationship given in (\ref{eq14}). The expression given (\ref{eq13}) can be obtained from
 \end{proof}
\section{Bivariate reliability concepts in reversed time}
\par Reversed hazard rate finds applications mainly in estimating the survival function of left censored data.  Gupta et al. (1998) proposed the proportional reversed hazard model. The concept of stochastic ordering based on reversed hazard rate is very popular in reliability theory (Finkelstein (2002), Nanda et al. (2003)).  Reversed hazard rate in bivariate set up is first introduced by Roy (2002). Here we define bivariate reversed hazard rate as a vector using quantile curves. The folowing definitions are for the direction $\varepsilon _{ -  - } $.

\begin{definition}Quantile curves based bivariate reversed hazard rate can be defined as the vector
\begin{equation*}
  \underset{\raise0.3em\hbox{$\smash{\scriptscriptstyle-}$}}{r} _{\varepsilon _{ -  - } } (u,P_{X}) = \left( {r_1 (u),r_2 (P_{X})} \right),
\end{equation*}
where
\begin{equation*}
  r_1 (u)=\frac{1}{u Q'_{X}(u)}\quad \textit{
and}\quad r_2 (P_{X})=\frac{1}{P_{X}\phi'(P_{X})}.
\end{equation*}
\end{definition}
\noindent The  $\underset{\raise0.3em\hbox{$\smash{\scriptscriptstyle-}$}}{r} _{\varepsilon _{ -  - } } (u,P_{X})$ uniquely determines the underlying quantile curve as
$$ Q_{\textbf{X}} (p,\varepsilon _{ - \, - } ) = \left( Q_{X} (u),\phi(P_{X}) \right),$$
where
\begin{equation*}
  Q_{X}(u)=\int_{0}^{u}{\frac{dz}{zh_{1}(z)}}\quad
and \quad
\phi(P_{X})=\int_{0}^{P_{X}}{\frac{dz}{zh_{2}(z)}}.
\end{equation*}
\par Nair and Asha (2008) defined the reversed mean residual life in bivariate set up as a vector reversed residual lives. The quantile curve based bivariate revered mean residual life is defined as
$$
\underset{\raise0.3em\hbox{$\smash{\scriptscriptstyle-}$}}{\eta} _{\varepsilon _{ -  - } } (u,P_{X}) = \left( {\eta_1 (u),\eta_2 (P_{X})} \right),
$$
where
$$
\eta_1 (u)=  Q_{X} (u)-\frac{1}{{ u}}\int_0^{u} {Q_{X} (z)dz }
$$
and
$$
\eta_2 (P_{X})=  \phi (P_{X})-\frac{1}{{P_{X}}}\int_{0}^{P_{X}} {\phi (z)dz }.
$$
Therefore $\underset{\raise0.3em\hbox{$\smash{\scriptscriptstyle-}$}}{\eta} _{\varepsilon _{ -  - } } (u,P_{X})$ gives the quantile curve as
$$ Q_{\textbf{X}} (p,\varepsilon _{ - \, - } ) = \left( Q_{X} (u),\phi(P_{X}) \right),$$
where
\begin{equation}
Q_{X}(u)=\eta_{1}(u)+\int_{0}^{u}\frac{\eta_1(z)}{z}dz
\end{equation}
 and
 \begin{equation}
\phi(P_{X})=\eta_{2}(P_{X})+\int_{0}^{P_{X}}\frac{\eta_2(z)}{z}dz.
\end{equation}
Unlike quantile curve  based bivariate mean residual life, means of the conditional random variables $Y_{X}$ or $X_{Y}$ are  not necessary for finding the quantile curve from quantile curve based bivariate  reversed mean residual life.\\

\section{Conclusions}
In this paper, we proposed a theoretical framework for extending the univariate quantile based reliability concepts to bivariate set up  in terms of quantile curves (level curves). We proposed quantile curves based bivariate hazard rate and bivariate mean residual life function and studied their uniqueness properties to determine the underlying quantile curve. A relationship between quantile curves based bivariate hazard rate and bivariate mean residual life function was also derived. We study the quantile curves based bivariate reliability concepts in reversed time setup as well. In essence, the work done in this paper can be considered as a theoretical foundation for developing reliability concepts in bivariate setup based on quantile functions of conditional distributions. \\

\onehalfspace
\textbf{References}

\begin{enumerate}
\item{}  Av{\'e}rous, J. and Meste, M. (1997), Median balls: an extension of the interquantile intervals to multivariate distributions, {\em Journal of Multivariate Analysis}, 63, 222-241.

\item{} Belzunce, F., Casta{\~n}o, A., Olvera-Cervantes, A. and Su{\'a}rez-Llorens, A. (2007), Quantile curves and dependence structure for bivariate distributions, {\em Computational Statistics and Data Analysis},
51,  5112-5129.


\item{} Fern{\'a}ndez-Ponce, J. M. and Su{\'a}rez-Llorens, A. (2003), A multivariate dispersion ordering based on quantile more widely separated, {\em Journal of Multivariate Analysis}, 85,  40-53.
\item{}  Finkelstein, M. S. (2002), On the reversed hazard rate, {\em Reliability Engineering and System Safety}, 78, 71–75.
\item Gilchrist, W. (2000), {\em Statistical modeling  with quantile functions}, CRC Press, Florida.
\item{} Gupta, R. C., Gupta, R. D. and Gupta, P. L. (1998),  Modeling failure time data by Lehman alternatives, {\em Communication in Statistics-Theory and  Methods}, 27, 887–904.

\item{} Koshevoy, G. and Mosler, K. (1997a), Zonoid trimming for multivariate distributions, {\em Annals of Statistics}, 25,  1998-2017.

\item{} Kulkarni, H. V. and Rattihalli, R. N. (1996), On Characterization of Bivariate Mean Residual Life Function, {\em IEEE Transactions on Reliability}, 38, 362-364.

\item{} Mass{\'e}, J. C. and Theodorescu, R. (1994), Halfplane trimming for bivariate distributions, {\em Journal of Multivariate Analysis}, 48, 188-202.
\item{} Mosler, K.  (2002), {\em Multivariate dispersion, central regions and depth: the lift zonoid approach},  Springer-Verlag, New York.
\item{} Nair, N. U. and  Asha, G. (2008), Some characterizations based on bivariate reversed mean residual life, {\em  ProbStat Forum}, 1, 1–14.
\item{}  Nair, K. R. M. and Nair, N. U. (1989), Bivariate Mean Residual Life, {\em IEEE Transactions on Reliability}, 38, 362-364.
\item{} Nair, N. U., Sankaran, P. G. and Balakrishnan, N. (2013), {\em Quantile  based reliability analysis}, Birkh{\"a}user, Basel.
\item{} Nanda, A. K., Singh, H., Neeraj, M. and Prasanta, P. (2003), Reliability properties of reversed residual lifetime, {\em Communication in Statistics-Theory and methods}, 32, 2031–2042.
\item{} Nolan, D. (1992), Asymptotics for multivariate trimming, {\em Stochastic processes and their applications}, 42, 157-169.
\item{} Roy, D. (2002), A characterization of model approach for generating bivariate life distributions
using reversed hazard rates, {\em Journal of the Japan statistical society}, 32, 239–245.


\item{} Sreelakshmi N. (2017), An introduction to copula based bivariate reliability concepts, {\em Communication in Statistics-Theory and methods}, accepted. \\
 DOI: 10.1080/ 03610926.2017.1316396.

\item{} Tukey, J. W. (1977), {\em Exploratory data analysis}, Addison-Wesley,  Princeton.
\end{enumerate}

\end{document}